\pdfoutput=1
\documentclass[a4paper,11pt]{amsart}
\usepackage{amsmath}
\usepackage{amsaddr}
\usepackage{amssymb}
\usepackage{hyperref}
\usepackage{cite}
\usepackage{graphicx}
\usepackage{subfig}

\newcommand{\Tr}{\mathrm{Tr}}
\newcommand{\tr}{\Tr}

\newcommand{\ket}[1]{\ensuremath{|#1\rangle}}
\newcommand{\bra}[1]{\ensuremath{\langle#1|}}
\newcommand{\ketbra}[2]{\ensuremath{\ket{#1}\bra{#2}}}
\newcommand{\proj}[1]{\ensuremath{\ket{#1}\bra{#1}}}
\newcommand{\braket}[2]{\ensuremath{\langle{#1}|{#2}\rangle}}
\newcommand{\kett}[1]{\ensuremath{|#1\rangle\rangle}}

\newcommand{\LL}{\mathcal{L}}

\newcommand{\1}{{\rm 1\hspace{-0.9mm}l}}
\newcommand{\Id}{\1}
\newcommand{\ii}{\mathrm{i}}
\newcommand{\dd}{\mathrm{d}}

\newcommand{\transpose}{\ensuremath{\intercal}}

\newcommand{\N}{\mathbb{N}}
\newcommand{\Z}{\mathbb{Z}}

\newcommand{\R}{\mathbb R}

\newcommand{\ie}{\emph{i.e.\/}}

\newcommand{\etal}{\emph{et al.}}

\newtheorem{theorem}{Theorem}
\newtheorem{lemma}[theorem]{Lemma}

\newtheorem{proposition}[theorem]{Proposition}
\newtheorem{remark}[theorem]{Remark}

\title{Properties of quantum stochastic walks from the asymptotic scaling 
exponent}

\author{Krzysztof Domino \and Adam Glos \and Mateusz Ostaszewski \and {\L}ukasz 
Pawela$^*$ \and Przemys{\l}aw Sadowski}
\address{Institute of Theoretical and Applied Informatics, 
Polish Academy of Sciences, Ba{\l}tycka 5, 44-100 Gliwice, Poland}
\email{lpawela@iitis.pl}

\begin{document}
\maketitle

\begin{abstract}
This work focuses on the study of quantum stochastic walks, which are a
generalization of coherent, \ie{} unitary quantum walks. Our main goal is to
present a measure of a coherence of the walk. To this end, we utilize the
asymptotic scaling exponent of the second moment of the walk \ie{} of the mean
squared distance covered by a walk. As the quantum stochastic walk model
encompasses both classical random walks and quantum walks, we are interested
how the continuous change from one regime to the other influences the
asymptotic scaling exponent. Moreover this model allows for behavior which is
not found in any of the previously mentioned model -- the model with global
dissipation. We derive the probability distribution for the walker, and
determine the asymptotic scaling exponent analytically, showing that ballistic
regime of the walk is maintained even at large dissipation strength.
\end{abstract}

\section{Introduction} The examination of the interaction of the quantum system
with the environment is important in quantum computing, where the environmental
interaction may destroy coherence and disturb quantum computation
\cite{gonis2001decoherence, amin2009decoherence}. Stochastic walks governed by
the Gorini-Kossakowski-Sudarshan-Lindblad (GKSL)~\cite{kossakowski1972quantum,
lindblad1976generators, gorini1976completely} equation can be used to examine
the quantum system that interacts with an environment
\cite{whitfield2010quantum}. Quantum walks~\cite{venegas2012quantum} have been
widely applied in quantum information theory for studying quantum
transport~\cite{robens2015ideal}, interference~\cite{stefanak2010interference},
entanglement~\cite{vieira2013dynamically},
measurement~\cite{kurzynski2013quantum}, entropy
production~\cite{kollar2014entropy}, topological
phases~\cite{kitagawa2010exploring}, gauge theories~\cite{arnault2016quantum}
and relativistic quantum mechanics~\cite{chandrashekar2010relationship} to name
a few. There are also examples of application of the quantum stochastic walk
model outside the field of quantum information, such as the description of
production of radical pairs in biological structures~\cite{chia2016coherent}.

In order to measure the coherence of such walks, we use the asymptotic scaling
exponent $\alpha$ of the mean squared distance covered by the walk
\cite{metzler2014anomalous,spiechowicz2016transient,spiechowicz2015diffusion,klages2008anomalous}
 that fulfills, for large times $t$, the following scaling law of the second 
moment of frequency distribution of the walker position $\mu_2(t) \propto 
t^{\alpha}$. It is important to note here, that for the quantum walk on a 
closed system, the walk is ballistic in the sense that $\alpha = 2$ 
\cite{kempe2003quantum}. However if we introduce a local interaction with an 
environment, the normal diffusion regime $\alpha = 1$ is reconstructed for 
large $t$ \cite{caldeira1983path}. Remarkably, there are some quantum walks, 
where regardless of the environment interaction, the ballistic regime ($\alpha 
= 2$) is maintained \cite{prokof2006decoherence}.

\subsection{The model}

The quantum stochastic walk model was introduced by Whitfield
\etal~\cite{whitfield2010quantum}. It is based on the GKSL master equation
\begin{equation}
\dot{\rho} = -\ii [H, \rho] + \sum_k \left( \gamma_k L_k \rho L_k^\dagger - 
\frac12 
\gamma_k
\{L_k ^\dagger L_k, \rho\} \right),
\end{equation}
where we have set $\hbar=1$ and $\{A, B\}$ denotes the anticommutator. In our
study we set all $\gamma_k =1$. Furthermore, we introduce an additional
parameter $\omega \in [0,1]$, which allows for a smooth transition from purely
coherent to purely dissipative evolution
\begin{equation}
\dot{\rho} = -\ii (1-\omega) [H, \rho] + \omega\sum_k \left( L_k \rho 
L_k^\dagger - 
\frac12 \{L_k ^\dagger L_k, \rho\} \right) \label{eq:stochastic}.
\end{equation}
As usual, we will take the Hamiltonian $H$ to be the adjacency matrix $A$ of
the underlying graph. Thus, in the case of a walk on a line, this matrix has a
very simple tridiagonal structure: $\bra{i}A\ket{i+1} = \bra{i+1}A\ket{i} = 1$
and $\bra{i}A\ket{j} = 0$ otherwise. After integrating
Eq.~\eqref{eq:stochastic} we obtain
\begin{equation}
\kett{\rho(t)} = S_\omega^t\kett{\rho(0)},\label{eq:evolution}
\end{equation}
where $\kett{\rho}$ denotes the vectorization of $\rho$ and is defined as a
linear operator for which
$\kett{\ketbra{i}{j}}=\ket{i}\ket{j}$. The operator $S_\omega^t$ is
\begin{equation}
S_\omega^t = \exp\left[
-\ii(1-\omega) t (H\otimes\1-\1\otimes \overline{H})
+
 \omega t \sum_k L_k \otimes \overline{L}_k 
-\frac{1}{2}L_k^\dagger L_k\otimes\1-\frac{1}{2}\1\otimes 
L_k^\transpose \overline{L}_k 
\right].\label{eq:integrated}
\end{equation}
Sometimes, for simplicity, we will write $\rho(t) = S_\omega^t(\rho(0))$ 
instead of Eq.~\eqref{eq:evolution}.

The dissipative part of Eq.~\eqref{eq:stochastic} allows us to recover a 
classical behavior of the walk. To this end, we choose the following set of 
operators

\begin{equation}
\LL = \left\{L | L = \frac{1}{\deg(m)} \ketbra{n}{m} , (m,n)\in E
\right\},\label{eq:classical}
\end{equation}
where $E$ is the collection of edges of the graph and $\deg(x)$
denotes the degree of a vertex. In the case of a walk on a line this reduces to
\begin{equation}
\LL = \left\{L | L = \frac12 \ketbra{i}{i+1} \vee L = \frac12 
\ketbra{i+1}{i}, i\in \Z \right\}.
\end{equation}

In this case the Lindblad operators model a local continuous quantum
measurement. Invoking continuous measurement
theory~\cite{jacobs2006straightforward}, we may interpret these non-hermitian
Lindblad terms as energy damping.

Now let us move to the global environment interaction case. This case is
interesting, since some coherence may be carried via an environment interaction,
sustaining the quantum walk regime even for relatively large environmental
interaction strength. To study this case, we use a single dissipation operator
\begin{equation}
L_S = \sum_{L\in \LL} L.\label{eq:sum}
\end{equation}
Obviously we have $L_S=L^\dagger_S$. This operator allows us to achieve dynamics
which cannot be seen neither in the classical, nor in the quantum limit. In the
case of a walk on a line segment of length $n$, the operator $L_S$ has a
slightly different form
\begin{equation}
\bra{i}L_{S}\ket{i+1} = \bra{i+1}L_{S}\ket{i} = \frac12 \;\; \mathrm{and} \;\; 
\bra{i}L_{S}\ket{j}=0 
\; \;
\mathrm{otherwise}.\label{eq:lin-line-segment}
\end{equation}
Hence, in this case, the operator $L_S$ is given by a subset of $n$ rows and 
the corresponding $n$ columns of the operator on a line. This operator 
represents a continuous position measurement.

\section{Probability distribution of stochastic quantum walks}
Here we will derive the probability distributions for stochastic quantum walks 
on a line segment and on an infinite line.
\subsection{Stochastic quantum walk on a line segment}
In this section we derive the probability distribution of the quantum
stochastic walk on a line segment of length $n$ for the case with a single
dissipation operator, $L_S$. We have the following Theorem

\begin{theorem}\label{thr:probOnSegment}
Given a stochastic walk on a line segment of length $n$ with the dissipation 
operator defined in Eq.~\eqref{eq:lin-line-segment} with some initial state 
$\rho(0)=\proj{l}$ where $l\in\{1,\dots,n\}$, the diagonal part of $\rho(t)$ is 
given by
\begin{equation}
\begin{split}
\bra{k}\rho(t)\ket{k}&=\left(\frac{2}{n+1}\right)^2 \sum\limits_{i,j=1}^n 
\sin\left(\frac{ki\pi}{n+1}\right)
\sin\left(\frac{kj\pi}{n+1}\right)\sin\left(\frac{li\pi}{n+1}\right)
\sin\left(\frac{lj\pi}{n+1}\right)\times\\
&\times 
\exp\left[-\frac 
t2\omega(\lambda_i - \lambda_j)^2\right]
\exp\left[ 2\ii t (1-\omega) (\lambda_i-\lambda_j)) 
\right],
\end{split}
\end{equation}
where
\begin{equation}
\lambda_i = \cos\left( \frac{i\pi}{n+1} \right).
\end{equation}
\end{theorem}
\begin{proof}
Using Eq.~\eqref{eq:integrated} we have
\begin{equation}
\label{eq:integrate_lind}
S_\omega^t = \exp \left[t\omega \left( L_{S}\otimes L_{S} -\frac{1}{2} 
L_{S}^2\otimes\Id -\frac{1}{2}\Id\otimes L_{S}^2\right) - \ii t 
(1-\omega)\left( H \otimes \1 - \1 \otimes H \right)\right].
\end{equation}
Now we note that in the case of the walk on a line segment, we have $2L_S=H$
and $[L_S\otimes L_S, L_S^m \otimes \1] = 0$. Hence, the eigenvectors of
$S^t_{\omega}$ are the same as the eigenvectors of $L_S \otimes L_S$. It is
straightforward to check that
\begin{equation}
S_\omega^t=\sum_{i,j}\exp(-\omega\frac{t}{2}(\lambda_i-\lambda_j)^2) \exp(2 \ii 
t (1-\omega)(\lambda_i - \lambda_j)) 
\ketbra{\lambda_i,\lambda_j}{\lambda_i,\lambda_j},
\end{equation}
where $\lambda_i$  and $\ket{\lambda_i}$ denote the eigenvalues and
eigenvectors of $L_S$. As $L_S$ is a tridiagonal Toeplitz matrix its
eigenvalues are given by~\cite{pasquini2006tridiagonal}
\begin{equation}
\lambda_j=\cos\Big(\frac{j\pi}{n+1}\Big),
\end{equation}
where $1\leq j \leq n$. Furthermore the elements of the eigenvectors matrix are
\begin{equation}
\braket j {\lambda_i}= \sqrt{\frac{2}{n+1}}\sin\left( 
\frac{ij\pi}{n+1} \right) = \braket{i}{\lambda_j}.
\end{equation}
From this we get that the elements of $S_\omega^t$ in the computational basis 
are
\begin{equation}
\begin{split}
\bra{\gamma,\delta}S^t_{\omega}\ket{\kappa, \beta}&=\sum_{i,j=1}^n 
\braket{i}{\lambda_{\kappa}}\braket{j}{\lambda_\beta}  
\braket{i}{\lambda_\gamma} \braket{j}{\lambda_\delta} 
\exp(-\omega\frac{t}{2}(\lambda_i-\lambda_j)^2) \exp(2\ii t 
(1-\omega)(\lambda_i - \lambda_j)) =\\
&=\left(\frac{2}{n+1}\right)^2 \sum_{i,j=1}^n 
\sin\left( \frac{\kappa i\pi}{n+1} \right)
\sin\left( \frac{\beta j\pi}{n+1} \right)
\sin\left( \frac{\gamma i\pi}{n+1} \right)
\sin\left( \frac{\delta j\pi}{n+1} \right) \times \\
&\phantom{=\ }\times \exp(-\omega\frac{t}{2}(\lambda_i-\lambda_j)^2) \exp(2\ii 
t 
(1-\omega)(\lambda_i - \lambda_j)).
\end{split}
\end{equation}
Putting $\kappa=\beta=k$ and $\gamma=\delta=l$ we recover the desired result.
\end{proof}

Next, we will study the behavior of the walk for large time $t$. We have the 
following result.
\begin{proposition}\label{th:asymptotic-prob}
Suppose we have a stochastic walk on a line segment of length $n$ described by
Eq.~(\ref{eq:stochastic}). If  $\omega\in(0,1]$  and we start in some
initial state $\rho(0) = \ketbra{l}{l}$ for $l\in\{1,\dots,n\}$, than for $t 
\to \infty$, $\rho(t)$
converges to some $\rho_\infty$ such that:
\begin{enumerate}
\item If $n$ is odd and $l=\frac{n+1}{2}$ we have
\begin{equation}
\bra{k}\rho\ket{k} = \begin{cases}
\frac{2}{n+1} ,& k=l,\\
\frac{1}{n+1}, & k\neq l.
\end{cases}
\end{equation}

\item Otherwise, when $l \neq \frac{n+1}{2}$ we have
\begin{equation}
\bra{k}\rho\ket{k} = \begin{cases}
\frac{3}{2(n+1)} ,& k=l \textrm{ or }k=n+1-l,\\
\frac{1}{n+1}, & \textrm{otherwise.}
\end{cases} 
\end{equation}
\end{enumerate}
\end{proposition}
The proof of Proposition~\ref{th:asymptotic-prob} is presented in 
Appendix~\ref{app:proof-th-asymptotic}.

\subsection{Quantum stochastic walk on a line}

In this section we derive the probability distribution for a quantum stochastic 
walk on a line as a function of the time $t$. We have the following theorem

\begin{theorem}\label{th:stochastic-prob}
For a stochastic walk on a line with an initial state $\rho(0) = \proj{0}$, the 
diagonal part of $\rho(t)$ is given by
\begin{equation}
\begin{split}
\bra{k}\rho(t)\ket{k}&=\frac{1}{4\pi^2}\int_{-\pi}^\pi\int_{-\pi}^\pi 
\cos(kx)\cos(ky)\exp\left[-\omega\frac{t}{2}(\cos(x)-\cos(y))^2\right] \\ 
&\phantom{=\ } \times \exp\left[2 \ii t (1-\omega)(\cos(x)-\cos(y))\right] \dd 
x\dd y.
\end{split}
\end{equation}
\end{theorem}

\begin{proof}
In the case of a walk on a line segment $[-n, \ldots, n]$, after using simple
trigonometric identities and assuming $n=4m+3$, $m \in \Z$, from
Theorem~\ref{thr:probOnSegment} we get
\begin{equation}
\begin{split}
\bra{k}\rho(t)\ket{k}=&\frac{1}{(n+1)^2} 
\sum\limits_{i,j=-n}^n
\sin\left (\frac{ki\pi}{2(n+1)}+\frac{k\pi}{2}+\frac{i\pi}{2}\right )
\sin\left (\frac{kj\pi}{2(n+1)}+\frac{k\pi}{2}+\frac{j\pi}{2}\right )
\times \\
&\times\sin\frac{i\pi}{2}	\sin\frac{j\pi}{2}
\exp\left[-\omega\frac{t}{2}(\sin\frac{\pi i}{2(n+1)}-\sin\frac{\pi 
j}{2(n+1)})^2\right]\times\\
&\times\exp\left[2\ii (1-\omega) t(\sin\frac{\pi i}{2(n+1)}-\sin\frac{\pi 
j}{2(n+1)})\right].
\end{split}
\end{equation}
Note, that for even $i$ or even $j$, the elements under the sum are equal to 
zero. We get
\begin{equation}
\begin{split}
\bra{k}\rho(t)\ket{k}&=\frac{1}{(n+1)^2} 
\sum\limits_{i,j \;\; \mathrm{odd}}
\cos\left (\frac{ki\pi}{2(n+1)}+\frac{k\pi}{2}\right )
\cos\left (\frac{kj\pi}{2(n+1)}+\frac{k\pi}{2}\right )
\times \\
&\phantom{\ =}\times \exp\left[-\omega\frac{t}{2}\left(\sin\frac{\pi 
i}{2(n+1)}-\sin\frac{\pi j}{2(n+1)}\right)^2\right]\times \\
		&\phantom{\ =}\times \exp\left[2\ii (1-\omega) t\left(\sin\frac{\pi 
		i}{2(n+1)}-\sin\frac{\pi 
				j}{2(n+1)}\right)\right].
\end{split}
\end{equation}
Note, that the formula above is $1/4$ of the Riemann sum of the function
\begin{equation}
\begin{split}
f(x)&=\cos\left (\frac{k\pi x}{2}+\frac{k\pi}{2}\right )\cos\left (\frac{k\pi 
	y}{2}+\frac{k\pi}{2}\right )\exp\left[-\omega\frac{t}{2}(\sin\frac{\pi 
		x}{2}-\sin\frac{\pi 
		y}{2})^2\right]\times \\
		&\times \exp\left[2\ii(1-\omega)t(\sin\frac{\pi 
				x}{2}-\sin\frac{\pi 
				y}{2})\right]
\end{split}
\end{equation}
over the square $[-1,1]\times [-1,1]$ when we divide the region into equal
squares. Hence, taking the limit $n \to \infty$ we get
\begin{equation}
\begin{split}
\bra{k}\rho(t)\ket{k} = & \frac{1}{4}\int_{-1}^1\int_{-1}^1 	\cos\left 
(\frac{k\pi 
	x}{2}+\frac{k\pi}{2}\right 
)\cos\left (\frac{k\pi 
	y}{2}+\frac{k\pi}{2}\right )\exp\left[-\omega\frac{t}{2}(\sin\frac{\pi 
		x}{2}-\sin\frac{\pi 
		y}{2})^2\right] \times \\
&		\times\exp\left[2\ii (1-\omega)t(\sin\frac{\pi 
				x}{2}-\sin\frac{\pi 
				y}{2})\right]\dd x \dd y.
\end{split}
\end{equation}
After substituting $u=\frac{(x+1)\pi}{2}$ and $v=\frac{(y+1)\pi}{2}$ we have
\begin{equation}
\begin{split}
\bra{k}\rho(t)\ket{k}=&\frac{1}{\pi^2}\int_{0}^\pi\int_{0}^\pi \cos(ku)
\cos(kv) \exp\left[-\omega\frac{t}{2}(\cos u-\cos v)^2\right]\times\\ & \times
\exp\left[2\ii (1-\omega)t(\cos u-\cos v)\right] \dd u \dd v.
\end{split}
\end{equation}
By symmetry with respect to $x=0$ and $y=0$ we obtain the result.
\end{proof}

\section{Purity}
Now we consider the purity of the state in a quantum stochastic walk on a line
segment as a function of the parameter $\omega$. For the set of dissipation
operators given by Eq.~\eqref{eq:classical} the behavior is not monotonic in
$\omega$. On the other hand, for the single dissipation operator given by
Eq.~\eqref{eq:lin-line-segment}, the purity is a non-increasing function of
$\omega$.

\begin{proposition}\label{th:purity}
Let us assume $0 < \omega_1 < \omega_2 < 1$. Then, given $S_{\omega_1}^t$,
$S_{\omega_2}^t$ defined as in Eq.~\eqref{eq:integrated} and assuming a single
dissipation operator $L_S$ given by Eq.~\eqref{eq:sum}, the following holds
\begin{equation}
\tr \left[ S_{\omega_1}^t(\rho)\right]^2 \geq \tr \left[ 
S_{\omega_2}^t(\rho)\right]^2
\end{equation}
for all $\rho$.
\end{proposition}
\noindent
In order to prove this proposition, we will need the following lemma
\begin{lemma}
\label{th:majorization_for_unital}
Suppose $\rho=\Phi(\sigma)$, where $\Phi$ is a unital channel, \ie\ 
$\Phi(\1)=\1$. Let $\lambda(\rho)$ denote the eigenvalues of $\rho$ in a 
decreasing order. Then $\lambda(\rho)\prec \lambda(\sigma)$, where $\prec$ 
denotes the majorization relation..
\end{lemma}

\begin{proof}[Proof of Proposition~\ref{th:purity}]
We assume $0 < \omega_1 < \omega_2 < 1$, so $\omega_2 = \omega_1 + 
\varepsilon$. Again, using the fact that $L_S$ and $H$ commute we obtain
\begin{equation}
S^t_{\omega_2}=\hat{S}^t_\varepsilon S^t_{\omega_1},
\end{equation}
where
\begin{equation}
\hat{S}^t_\varepsilon=\exp\left[t\varepsilon(L_{S}\otimes L_{S} 
	-\frac{1}{2}L_{S}^2\otimes\Id
	-\frac{1}{2}\Id\otimes L_{S}^2) +t\varepsilon\ii (H \otimes\Id - 
	\Id\otimes H )\right].
\end{equation}
Now, noticing that $\hat{S}^t_\varepsilon$ is unital, the Proposition follows 
from 
Lemma~\ref{th:majorization_for_unital}.
\end{proof}
\begin{figure}[h!]
\subfloat[Stochastic case]{\includegraphics{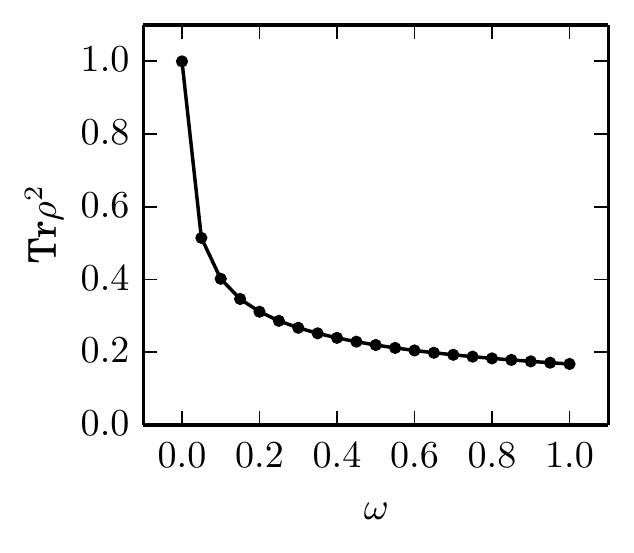}}
\subfloat[Classical case]{\includegraphics{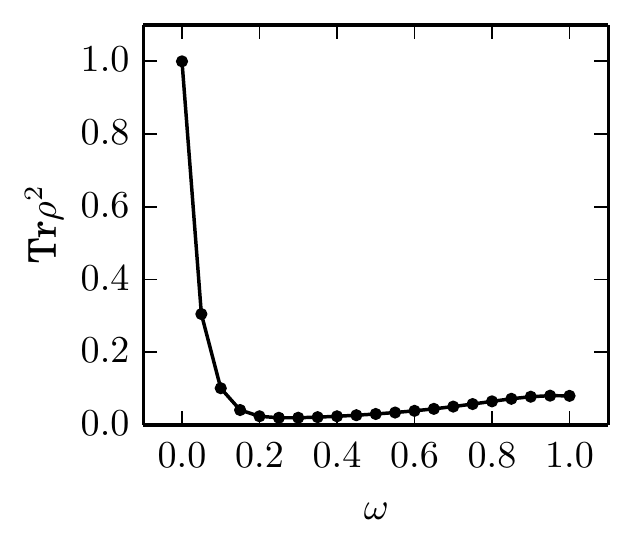}} 

\caption{Purity of final state for stochastic walk on a line segment of length 
400, with time 25 and the initial state is $\ketbra{200}{200}$.
}\label{f::piurity}
\end{figure}

Analyzing Fig.~\ref{f::piurity} one can see that the purity in the global 
dissipation case, decreases rapidly 
with $\omega$, and has sharp maximum at $\omega = 0$. The conclusion is that, 
if the ballistic regime, $\alpha=2$, is sustained for strong environment 
interaction, $\omega \rightarrow 1$, it is not caused by a high purity.

\section{The asymptotic scaling exponent approach}

In this section we investigate the asymptotic scaling of the second central
moment of the frequency distribution of the walker position
\cite{havlin1982fractal, haus1987diffusion, havlin1987diffusion}.

\begin{remark}
The asymptotic scaling exponent distinguishes between different types of 
diffusion process 
\cite{metzler2014anomalous,spiechowicz2016transient,spiechowicz2015diffusion,
	klages2008anomalous}:
\begin{itemize}
\item if $\alpha < 1$ the process is sub-diffusive,
\item if $\alpha = 1$ the process obeys a normal diffusion regime,
\item if $1 < \alpha < 2$ the process is super-diffusive,
\item if $\alpha = 2$ the process obeys a ballistic diffusion regime.
\end{itemize} 
\end{remark}

For the standard Markovian random walk on a uniform lattice the asymptotic 
scaling exponent
value is $\alpha = 1$ \cite{havlin1987diffusion}. For the quantum walk on
such a lattice it is $\alpha = 2$ \cite{grimmett2004weak}.

Consider now the general stochastic walk which is a combination of the quantum
walk, where $\alpha = 2$, and purely dissipative stochastic walk, where $\alpha 
= 1$ no matter which of the studied dissipation operators are chosen. We 
use the asymptotic scaling exponent to evaluate which regime, classical or 
quantum, 
dominates the process for $\omega \in (0, 1)$.

Let us discuss now the procedure used to calculate the asymptotic scaling
exponent. The second moment of the probability distribution of the stochastic
walk at time $t$ is given by
\begin{equation}
\mu_2(t) =\sum_{k}k^2 \bra{k}\rho(t)\ket{k}.
\end{equation}
To compute the asymptotic scaling exponent we use the standard scaling relation 
~\cite{metzler2014anomalous,spiechowicz2016transient,spiechowicz2015diffusion,
	klages2008anomalous,haus1987diffusion}
\begin{equation}\label{eq:h1}
\mu_2(t) \propto t^{\alpha},
\end{equation}
We have the following results regarding the moments of the stochastic walk. 
First, we calculate the $m$\textsuperscript{th} central moment for a stochastic 
walk on a line.
\begin{proposition}\label{th:stochastic-moments}
For a stochastic walk on a line with an initial state $\rho(0) = \proj{0}$,
$\omega =1$ and $L_S$ as in Eq~\eqref{eq:sum}, the $m$-th central moment 
$\mu_m(t)$ is polynomial in $t$ for $m$
even and zero otherwise. Moreover for even $m$ we have
\begin{equation}
\lim_{t\to\infty}\frac{\mu_m(t)}{t^\frac{m}{2}}= \frac{m!}{\left( \frac m2 
\right)!8^\frac{m}{2}}\binom{m}{\frac{m}{2}}.
\end{equation}

\end{proposition}
The proof of Proposition~\ref{th:stochastic-moments} is given in Appendix
\ref{app:stochastic-moments}. From this, we get that the second central moment 
has the form $\mu_2(t)=\frac{1}{2}t$.

Now, we move to the case of arbitrary non-zero $\omega$. We study the leading 
term in the $m$\textsuperscript{th} central moment.

\begin{proposition}\label{th:stochastic-mixed-moments}
For a stochastic walk on a line with an initial state $\rho(0) = \proj{0}$ and 
$\omega \in(0,1)$, the 
$m$-th central moment $\mu_m(t)$ is polynomial in $t$ for even $m$ and zero 
otherwise. Moreover for even $m$ we have
\begin{equation}
\lim_{t\to\infty}\frac{\mu_m(t)}{t^m} = \binom{m}{\frac{m}{2}}(\omega-1)^m.
\end{equation}
\end{proposition}
\noindent
Now, we calculate the second central moment. We get the following result.
\begin{proposition}\label{th:stochastic-mixed-second-moment}
For a stochastic walk on a line with an initial state $\rho(0) = \proj{0}$ and 
$\omega\in(0,1]$ the 
second central moment is of the form
\begin{equation}\label{eq::2m}
\mu_2(t) = 2\left(\omega-1\right)^2t^2+\frac{\omega}{2}t.
\end{equation}
\end{proposition}
\noindent
Proofs of Propositions~\ref{th:stochastic-mixed-moments} and
\ref{th:stochastic-mixed-second-moment} are given in
Appendix~\ref{app:mixed-stochastic-second-moment}.

From these propositions, we get that the asymptotic scaling exponent in the 
case of the 
Lindblad operator as in Eq~\eqref{eq:sum} is:
\begin{enumerate}
\item if $\omega \neq 1$  (general stochastic walk case) we have $\alpha = 2$,
\item if $\omega = 1$ (purely dissipative case) we have $\alpha = 1.$
\end{enumerate}

At first glance, this results seems to be quite surprising as the walk exhibits 
ballistic behavior as long as we have a finite amount of a coherent evolution. 
This effect can be explained as follows. If we have $\omega < 1$, then the 
quantum effects dominate the system, despite the measurement. However, if we 
turn off the coherent evolution, then we get a behavior determined only by the 
random disturbances of the system due to the measurement. In this case we get a 
behavior resembling a classical random walk in the sense that $\alpha = 1$.

On the other, when we consider the local dissipation, we get the opposite
result. In this case we get that $\alpha = 2$ if $\omega=0$ and $\alpha < 2$ if 
$\omega > 0$.
This can be explained easily. As stated in the introduction, the Lindblad terms
in this case introduce energy damping to the evolution of the system. If we have
a system driven by a Hamiltonian $H$ and perform a continuous energy damping on
it, then for large times $t$, the system becomes completely classical.

\section{Conclusions}

In this work we studied the behavior of quantum stochastic walks on a line and
a line segments. We found an analytical formula for the probability
distribution of the walk in the case of a global dissipation operator,
equivalent to a continuous position measurement. To measure the quantum
coherence of the walk, we used the asymptotic approach. Hence we investigated
the scaling of the second moment of the frequency distribution, by means of the
asymptotic scaling exponent, $\alpha$. There are other known approaches to
quantifying quantumness in a quantum walk model. One of the is the amount of
needed corrections to the classical stationary state of the walk to recover a
time averaged quantum distribution~\cite{faccin2013degree}. The main 
difference between this result and ours is the fact that our result applies to 
arbitrary initial state, whereas the result from~\cite{faccin2013degree} 
requires that the initial state is mostly in the ground space of the walk's 
Hamiltonian.

In the case of the stochastic walk on a line, with the global dissipation
operator the ballistic propagation was demonstrated, since for large walk time,
the relation $\alpha = 2$ holds. Importantly, such ballistic propagation occurs
for each strength of the dissipation as long as the Hamiltonian term is present
in the master equation. This result is surprising and in contrary to the
asymptotic behavior of the stochastic walk with local dissipation, equivalent
to continuous energy damping. In such case even small dissipation leads to the
classical propagation of the walk and $\alpha = 1$ for large time of the walk.

If we analyze the asymptotic behavior, the global dissipation leads to a
quantum-like propagation. This result is interesting while constructing quantum
system that is supposed to keep the coherence regardless of an environmental
interaction. Such observation may be useful in a quantum computer development,
where maintaining quantum coherence is crucial.

\section*{Acknowledgements}
KD and {\L}P acknowledge the support of the National Science Centre, Poland 
under project number 2014/15/B/ST6/05204. PS acknowledges the support of the 
National Science Centre, Poland under project number 2013/11/N/ST6/03030. AG 
and MO were supported by the Polish Ministry of Science and Higher Education 
under project number IP 2014 031073.

\appendix

\section{Proof of 
Proposition~\ref{th:asymptotic-prob}}\label{app:proof-th-asymptotic}
First, we will need a couple of technical lemmas:
\begin{lemma}[\cite{gradshteyn2014table}]\label{lem:sinFourthPower}
For arbitrary $n\in 
\mathbb N$ and $x\in\R$ such that $x\neq \frac{k\pi}{2}$ for $k\in \Z$  
\begin{equation}
\sum_{i=1}^n \sin^4(ix) = \frac{1}{8}\left (3n-4\frac{\cos\left((n+1)x\right) 
\sin(nx)}{\sin (x)} + 
\frac{\cos 
\left(2(n+1)x \right) \sin(2nx) }{\sin(2x)} \right ).
\end{equation}
\end{lemma}

\begin{lemma}[\cite{gradshteyn2014table}]\label{lem:cosSum}
For arbitrary $n\in \mathbb N$ and $x\in\R$ such that $x 
\neq 2 k \pi$ for $k\in \Z$
\begin{equation}
\sum_{h=1}^n \cos(hx) = \frac{1}{2}\left(\frac{\sin 
\left(\frac{2n+1}{2}x\right)}{\sin \left(\frac{x}{2}\right)} -1\right).
\end{equation}
\end{lemma}

Now, we are ready to prove Proposition~\ref{th:asymptotic-prob}

\begin{proof}[Proof of Proposition~\ref{th:asymptotic-prob}]
By Theorem~\ref{thr:probOnSegment} we have
\begin{equation}
\lim_{t\to \infty} \rho_{kk}=\left(\frac{2}{n+1}\right)^2 \sum_{i=1}^{n} \sin^2 
\left(\frac{ki\pi}{n+1}\right) \sin^2 \left(\frac{li\pi}{n+1}\right).
\end{equation}
Now we analyze the sum above. Suppose $k=l$. Then we have
\begin{equation}
\rho_{kk}= \left(\frac{2}{n+1}\right)^2 \sum_{i=1}^{n} \sin^4
\left(\frac{ i k\pi}{n+1}\right).
\end{equation}
Suppose $k\neq \frac{n+1}{2}$. By Lemma~\ref{lem:sinFourthPower} the sum above 
equals
\begin{equation}
\begin{split}
\rho_{kk} &=\left(\frac{2}{n+1}\right)^2 \frac{1}{8}\left (3n-4\frac{\cos(\pi
k) \sin\left(\frac{n\pi k}{n+1}\right)}{ \sin \left(\frac{\pi k}{n+1}\right)} +
\frac{\cos (2\pi k) \sin \left(\frac{2\pi n k}{n+1}\right)}{\sin
\left(\frac{2\pi k}{n+1}\right)} \right )\\
&=\left(\frac{2}{n+1}\right)^2\frac{1}{8}\left (3n+4(-1)^k
(-1)^k\frac{\sin\left(\frac{\pi k}{n+1}\right)}{\sin\left(\frac{\pi
k}{n+1}\right)} - \frac{\sin\left(\frac{2\pi
k}{n+1}\right)}{\sin\left(\frac{2\pi k}{n+1}\right)} \right )\\
&=\left(\frac{2}{n+1}\right)^2\frac{1}{8}(3n+4-1) = \frac{3}{2(n+1)},
\end{split}
\end{equation}
since $\cos(k \pi)= (-1)^k$ and $\sin\frac{nk\pi}{n+1}=-(-1)^k\sin\frac{\pi 
k}{n+1}$ and similarly for second formula.
Now suppose $k=\frac{n+1}{2}$. Then
\begin{equation}
\begin{split}
\rho_{kk}&= \left(\frac{2}{n+1}\right)^2 \sum_{i=1}^{n} \sin^4
\left(\frac{ik\pi}{n+1}\right) = \left(\frac{2}{n+1}\right)^2 \sum_{i=1}^{n} 
\sin^4 \left(\frac{i\pi}{2}\right) \\
&=  \left(\frac{2}{n+1}\right)^2\frac{2n+1+(-1)^{n+1}}{4} = 
\left(\frac{2}{n+1}\right)^2\frac{n+1}{2} = \frac{2}{n+1},\\
\end{split}
\end{equation}
since $l=\frac{n+1}{2}$ implies $n+1$ is even.

Suppose we have $k+l=n+1 $. Then 
\begin{equation}
\begin{split}
\rho_{kk}&=\left(\frac{2}{n+1}\right)^2 \sum_{i=1}^{n} \sin^2 
\left(\frac{(n+1-k)i\pi}{n+1}\right) \sin^2 \left(\frac{ik\pi}{n+1}\right) \\
&=\left(\frac{2}{n+1}\right)^2 \sum_{i=1}^{n} \sin^4 
\left(\frac{ik\pi}{n+1}\right)=\frac{3}{2(n+1)}.
\end{split}
\end{equation}

Suppose now that $k\neq l$ and $k\neq n+1 -l$. Then by formula 
\begin{equation}
\sin^2(x)\sin^2(y)= \frac{1}{8} (\cos(2x-2y) +\cos(2x+2y)-2\cos(2x)-2\cos(2y) 
+2)
\end{equation}
and by Lemma~\ref{lem:cosSum} we have
\begin{equation}
\begin{split}
\rho_{kk} &= \left(\frac{2}{n+1}\right)^2  \frac{1}{16} \left (2+4n -2 
\frac{\sin\left(\frac{l(2n+1)\pi}{n+1}\right)
}{\sin\left(\frac{l\pi}{n+1}\right)} + 
\frac{\sin\left(\frac{(k-l)(2n+1)\pi}{n+1}\right)
}{\sin\left(\frac{(k-l)\pi}{n+1}\right)}  \right . \\ &\phantom {=\ }+ \left . 
\frac{\sin\left(\frac{(k+l)(2n+1)\pi}{n+1}\right)
}{\sin\left(\frac{(k+l)\pi}{n+1}\right)} -2 
\frac{\sin\left(\frac{k(2n+1)\pi}{n+1}\right)
}{\sin\left(\frac{k\pi}{n+1}\right)}\right) \\
&=\frac{1}{4(n+1)^2}\left (2+4n +2 
\frac{\sin\left(\frac{l\pi}{n+1}\right)
}{\sin\left(\frac{l\pi}{n+1}\right)} - 
\frac{\sin\left(\frac{(k-l)\pi}{n+1}\right)
}{\sin\left(\frac{(k-l)\pi}{n+1}\right)}  - 
\frac{\sin\left(\frac{(k+l)\pi}{n+1}\right)
}{\sin\left(\frac{(k+l)\pi}{n+1}\right)} +2 
\frac{\sin\left(\frac{k\pi}{n+1}\right)
}{\sin\left(\frac{k\pi}{n+1}\right)}\right ) \\
&=\frac{4(n+1)}{4(n+1)^2}=\frac{1}{n+1}.
\end{split}
\end{equation}
\end{proof}

\section{Proof of Proposition
\ref{th:stochastic-moments}}\label{app:stochastic-moments}

First we need a couple of technical lemmas:
\begin{lemma}[\cite{gradshteyn2014table}] \label{lem:sumOfPowerBinomial}
For arbitrary $\alpha\in \R$, $n,m\in \N$ such that $m\leq n$  we have
\begin{equation}
\sum_{k=0}^{n} (-1)^k (k-\alpha)^m \binom{n}{k} = \begin{cases}
0, & m<n,\\
(-1)^n n! & m = n.
\end{cases}
\end{equation}
\end{lemma}

\begin{lemma}[\cite{gradshteyn2014table}]\label{lem:sumDoubleBinomial}
For arbitrary $n,p\in\N$ such that $p\leq n$ we have
\begin{equation}
\sum_{k=0}^{n-p} \binom{n}{k}\binom{n}{p+k} = \binom{2n}{n-p}.
\end{equation}
\end{lemma}

\begin{lemma}\label{lem:cosIntegral}
For arbitrary $k,l\in \N$ we have
\begin{equation}
\int_{-\pi}^{\pi}\cos(kx) \left[\cos(x)\right]^l \dd x = \begin{cases}
\frac{2\pi}{2^{l-k}} \binom{l-k}{\frac{l-k}{2}} \prod_{i=0}^{k-1} 
\frac{l-i}{l+k-2i}, & l \geq k \textrm{ and }l=k\mod 2,\\
0, & \textrm{otherwise.}\\
\end{cases}
\end{equation}
\end{lemma}
\begin{proof}
Using the formula \cite{gradshteyn2014table}
\begin{equation}
\int \left[\cos(x)\right]^l \cos (kx) \dd x=
\frac{1}{l+k}\left[\left(\cos(x\right))^l 
\sin(kx)+l\int\left[\cos(x)\right]^{l-1}
\cos((k-1)x)\dd x\right]
\end{equation}
we obtain  
\begin{equation}\label{eq:integralRecurence}
\int_{-\pi}^\pi \left[\cos(x)\right]^l \cos (kx)\dd x = 
\frac{l}{l+k}\int_{-\pi}^\pi \left[\cos(x) \right]^{l-1}
\cos((k-1)x)\dd x.
\end{equation}
Moreover for arbitrary $l\in \N$ we have \cite{gradshteyn2014table}
\begin{equation}
\int \left[\cos (x) \right]^{2l} \dd x = \frac{1}{4^l}\binom{2l}{l}x +
\frac{2}{4^l} \sum_{k=0} ^{l-1} \binom{2l}{k}\frac{\sin ((2l-2k)x)}{2l-2k},
\end{equation}
which provides us the formula
\begin{equation}\label{eq:cosPowerIntegral}
\int_{-\pi}^\pi \left[ \cos (x)\right]^{2l} \dd x = 
\frac{2\pi}{4^l}\binom{2l}{l}.
\end{equation}

Suppose $l<k$. Then using Eq.~\eqref{eq:integralRecurence} we have 
\begin{equation}
\int_{-\pi}^\pi \left[\cos(x)\right]^l \cos (kx)\dd x = \prod_{i=0}^{l-1} 
\frac{l-i}{l+k-2i} 
\int_{-\pi}^{\pi} \cos((k-l)x) \dd x = 0.
\end{equation}
If $l\geq k$, then we obtain
\begin{equation}
\int_{-\pi}^\pi \left[ \cos(x)\right]^l \cos(kx) = \prod_{i=0}^{k-1} 
\frac{l-i}{l+k-2i} 
\int_{-\pi}^{\pi} \cos^{l-k} (x) \dd x.
\end{equation}
If $l - k$ is odd, then the integral equals 0. Otherwise using 
Eq.~\eqref{eq:cosPowerIntegral} we have
\begin{equation}
\int_{-\pi}^\pi\cos(kx)\cos^l(x) = 
\frac{2\pi}{2^{l-k}}\binom{l-k}{\frac{l-k}{2}}\prod_{i=0}^{k-1} 
\frac{l-i}{l+k-2i}.
\end{equation}
\end{proof}

\begin{proposition}\label{th:probabilityTaylorSeries}
For a stochastic walk on a line with an initial state $\rho(0) = \proj{0}$ and 
$\omega =1$, the $\rho_{kk}(t)$ has series representation with respect to 
$t$ of the form
\begin{equation}
\rho_{kk}(t)=  \sum_{n=|k|}^\infty 
\frac{(-1)^{n+k}}{8^n}\binom{2n}{n}\binom{2n}{n+k}\frac{t^n}{n!}.
\end{equation}
\end{proposition}
\begin{proof} Since the elements $\rho_{kk}(t)$ are symmetric with respect 
to $k=0$, we assume $k\geq0$. By Theorem~\ref{th:stochastic-prob} we have
\begin{equation}
\rho_{kk}(t)=\frac{1}{4\pi^2}\int_{-\pi}^\pi\int_{-\pi}^\pi 
\cos(kx)\cos(ky)\exp\left[-\frac{t}{2}(\cos(x)-\cos(y))^2\right]\dd x \dd y.
\end{equation}
Suppose we have the Taylor series representation $\rho_{kk}(t) = 
\sum_{n=0}^{\infty} 
\frac{A_{n,k}}{n!}t^n$. Then $A_{n,k}$ is of the form
\begin{equation}
\begin{split}
A_{n,k} &= \frac{(-1)^n}{4\cdot2^n\pi^2}\int_{-\pi}^\pi\int_{-\pi}^\pi 
\cos(kx)\cos(ky)\left 
(\cos(x)-\cos(y)\right)^{2n}\dd x \dd y \\
&= \frac{(-1)^n}{4\cdot2^n\pi^2}\sum_{l=0}^{2n}  
\binom{2n}{l}(-1)^l\int_{-\pi}^\pi\cos(kx) \left[\cos(x) \right]^l \dd 
x\int_{-\pi}^\pi 
\cos(ky) \left[\cos (y)\right]^{2n-l} \dd y.
\end{split}
\end{equation}

Let us define for simplicity 
\begin{equation}
A_{n,k,l} =
\binom{2n}{l}(-1)^l\int_{-\pi}^\pi\cos(kx)\left[\cos(x)\right]^l \dd 
x\int_{-\pi}^\pi 
\cos(ky) \left[\cos(y)\right]^{2n-l} \dd y.
\end{equation}
By Lemma~\ref{lem:cosIntegral} we have that $A_{n,k,l}$ is non-zero when $k-l$ 
is even and has the form
\begin{equation}
A_{n,k,l}  = 
\frac{(-1)^l 4\pi^2}{2^{2n-2k}} \binom{2n}{l} \binom{l-k}{\frac{l-k}{2}} 
\binom{2n-l-k}{n-\frac{l+k}{2}} 
\prod_{i=0}^{k-1} \frac{l-i}{l+k-2i} \frac{2n-l-i}{2n-l+k-2i}. 
\end{equation}
In particular, we note that for $n<k$ we have $A_{n,k}=0$. 

Again it is straightforward to find
\begin{equation}
A_{n,k,k} = \frac{(-1)^k 4\pi^2}{4^{n}} \binom{2n}{n}\binom{n}{k}
\end{equation}
and 
\begin{equation}
\frac{A_{n,k,l+2}}{A_{n,k,l}} = \frac{(n - \frac{l-k}{2})(n - 
\frac{l-k}{2})}{(\frac{l+k}{2}+1)(\frac{l-k}{2}+1)}.
\end{equation}
Note that we increment $l$ by two instead of one because of the assumption that
$l-k$ is even. One can verify, that the $A_{n,k,l}$ is of the form
\begin{equation}
A_{n,k,l} = 
\frac{(-1)^k4\pi^2}{4^n}\binom{2n}{n}\binom{n}{\frac{l+k}{2}} 
\binom{n}{\frac{l-k}{2}},
\end{equation}
and hence have
\begin{equation}
\begin{split}
A_{n,k} &=  \frac{(-1)^n}{4\cdot2^n\pi^2} 
\sum_{l\in\{k,k+2,\dots,2n-k\}}A_{n,k,l} \\
&= 
\frac{(-1)^{n+k}}{8^n}\binom{2n}{n}\sum_{l\in\{k,k+2,\dots,2n-k\}} 
\binom{n}{\frac{l+k}{2}}\binom{n}{\frac{l-k}{2}}\\
&= \frac{(-1)^{n+k}}{8^n}\binom{2n}{n}\sum_{l=0}^{n-k}
\binom{n}{l+k}\binom{n}{l}\\
&=\frac{(-1)^{n+k}}{8^n}\binom{2n}{n} \binom{2n}{n+k},
\end{split}
\end{equation}
where in the third line we change the indices range and in the last line we use 
Lemma~\ref{lem:sumDoubleBinomial}.
\end{proof}

Now we are ready to prove Proposition~\ref{th:stochastic-moments}.

\begin{proof}[Proof of Proposition \ref{th:stochastic-moments}]
Note that odd moments equals 0 by symmetry of the probability distribution. 
Suppose $m$ is even and $m>0$. Then by 
Proposition~\ref{th:probabilityTaylorSeries} 
we have
\begin{equation}
\begin{split}
\mu_m(t) &= \sum_{k=-\infty}^{\infty} k^m  \sum_{n=|k|}^\infty 
\frac{(-1)^{n+k}}{8^n}\binom{2n}{n}\binom{2n}{n+k}\frac{t^n}{n!}\\
&= 
\sum_{n=0}^{\infty}\frac{(-1)^n}{8^n}\binom{2n}{n}\frac{t^n}{n!}\sum_{k=-n}^{n}k^m
 (-1)^k\binom{2n}{n+k}\\
 &=\sum_{n=0}^{\infty}\frac{1}{8^n}\binom{2n}{n}\frac{t^n}{n!} 
 \sum_{k=0}^{2n}(k-n)^m   (-1)^{k}\binom{2n}{k}.
\end{split}
\end{equation}
By Lemma~\ref{lem:sumOfPowerBinomial} formula above can be simplified
\begin{equation}
\mu_m(t) = 
\sum_{n=1}^{\frac{m}{2}}\frac{1}{8^n}\binom{2n}{n}\frac{t^n}{n!}  
\sum_{k=0}^{2n}(k-n)^m   (-1)^{k}\binom{2n}{k},
\end{equation}
hence the $m$-th central moment is a polynomial of degree $\frac{m}{2}$ with 
respect to $t$. Moreover, the coefficient next to $t^\frac{m}{2}$ is
\begin{equation}
a_{\frac{m}{2}} = \frac{m!}{\left( \frac{m}{2} \right)! 
8^\frac{m}{2}}\binom{m}{\frac{m}{2}}. 
\end{equation}
\end{proof}

\section{Proof of Propositions~\ref{th:stochastic-mixed-moments} and 
\ref{th:stochastic-mixed-second-moment}} 
\label{app:mixed-stochastic-second-moment}

\begin{proof}[proof of Proposition~\ref{th:stochastic-mixed-moments}]
Case $\omega=1$ was already considered in Theorem~\ref{th:stochastic-moments}.
Suppose $\omega\in(0,1)$. As in Appendix~\ref{app:stochastic-moments}, we start
by finding the Taylor series representation for the probability distribution. 
If we denote
$\rho_{kk}(t)=\sum_{n=0}^{\infty}\frac{B_{n,k}}{n!}t^n$, then one can find that
$B_{n,k}$ is of the form
\begin{equation}
\begin{split}
B_{n,k} &=  \sum_{l=0}^{n } 
\binom{n}{l}\frac{1}{4\pi^2}\int_{-\pi}^\pi\int_{-\pi}^\pi 
\cos(kx)\cos(ky)\frac{(-\omega)^{n-l}}{2^{n-l}}\times  \\
&\phantom{\ =}\times (\cos(x)-\cos(y))^{2n-l} 2^l  \ii^l (1-\omega)^l \dd 
x\dd y .
\end{split}
\end{equation}
Since $\rho_{kk}(t) \in \R$, we can exclude the imaginary terms and we can
simplify the formula
 \begin{equation}
 \begin{split}
 B_{n,k} &=\sum_{l=0}^{\lfloor \frac{n}{2}\rfloor} \binom{n}{2l}
 8^{l}\omega^{n-2l}(1-\omega)^{2l} 
 \times \\
 &\phantom{\ =}\times \frac{(-1)^{n-l}}{2^{n-l} 
  4\pi^2}\int_{-\pi}^\pi\int_{-\pi}^\pi 
 \cos(kx)\cos(ky)(\cos(x)-\cos(y))^{2n-2l} \dd  x\dd y \\
 &=\sum_{l=0}^{\lfloor \frac{n}{2}\rfloor} \binom{n}{2l} 
  8^{l}\omega^{n-2l}(1-\omega)^{2l}
  A_{n-l,k}.
   \end{split}
 \end{equation}
 
From Appendix~\ref{app:stochastic-moments} we know, that $A_{n,k}$ is of the 
form
\begin{equation}
A_{n,k} = \begin{cases}
0, & |k|>n,\\
\frac{(-1)^{n+k}}{8^n} \binom{2n}{n}  \binom{2n}{n+k}, & |k|\leq n.
\end{cases}
\end{equation}
In our case we have the condition $|k|\leq n-\frac{l}{2}\leq n$. Hence we
conclude, that $B_{n,k}$ is of the form
\begin{equation}
B_{n,k}  = \frac{(-1)^{n+k}}{8^n}\sum_{l=0}^{\min 
(\lfloor\frac{n}{2}\rfloor,n-|k|)} 
\binom{n}{2l}
8^{2l}\omega^{n-2l}(1-\omega)^{2l}  
(-1)^{l}
\binom{2n-2l}{n-l}  
\binom{2n-2l}{n-l+k}.
\end{equation}

Therefore for even $m$ we have
\begin{equation}
\begin{split}
\mu_m(t)&= \sum_{k=-\infty}^{\infty} k^m  \sum_{n=0}^\infty 
B_{n,k}\frac{t^n}{n!}\\
&=\sum_{n=0}^{\infty}\frac{t^n}{n!}\sum_{k=-n}^{n}k^mB_{n,k}\\
&=\sum_{n=0}^{\infty}\frac{(-1)^nt^n}{8^nn!}\sum_{k=-n}^{n}k^m 
(-1)^k\sum_{l=0}^{\min 
(\lfloor\frac{n}{2}\rfloor,n-|k|)} 
\binom{n}{2l}
8^{2l}\omega^{n-2l}(1-\omega)^{2l}  
(-1)^{l}
\binom{2n-2l}{n-l}  
\binom{2n-2l}{n-l+k}\\
&=\sum_{n=0}^\infty \sum_{k=-n}^{n}\sum_{l=0}^{\min 
(\lfloor\frac{n}{2}\rfloor,n-|k|)} C_{n,l,t} \omega^{n-2l}(1-\omega)^{2l}(-1)^k 
k^m  \binom{2n-2l}{n-l+k},\\
&=\sum_{n=0}^\infty \sum_{l=0}^{\lfloor\frac{n}{2}\rfloor} 
C_{n,l,t} 
\omega^{n-2l}(1-\omega)^{2l} \sum_{k=-(n-l)}^{n-l} (-1)^k 
k^m  \binom{2n-2l}{n-l+k},\\
&=\sum_{n=0}^\infty \sum_{l=0}^{\lfloor\frac{n}{2}\rfloor} 
C_{n,l,t} 
\omega^{n-2l}(1-\omega)^{2l} (-1)^{n-l}\sum_{k=0}^{2n-2l} (-1)^k 
(k-(n-l))^m  \binom{2n-2l}{k},\label{eq:mth-moment-mixed}\\
\end{split}
\end{equation}
where 
\begin{equation}
C_{n,t,l}= \frac{(-1)^nt^n}{8^nn!}  8^{2l}(-1)^l 
\binom{n}{2l}\binom{2n-2l}{n-l} .
\end{equation}
Let us denote
\begin{equation}
\alpha_{2n-2l,m} = \sum_{k=0}^{2n-2l} (-1)^k 
(k-(n-l))^m  \binom{2n-2l}{k}.
\end{equation} 
From Lemma~\ref{lem:sumOfPowerBinomial} $\alpha_{2n-2l,m}$ is nonzero if 
$m\geq2n-2l\iff l\geq n-\frac{m}{2}$. Hence Eq.~\eqref{eq:mth-moment-mixed} 
can be simplified
\begin{equation}
\begin{split}
\mu_m(t) = \sum_{n=0}^\infty \sum_{l=n-\frac{m}{2}}^{\lfloor\frac{n}{2}\rfloor} 
C_{n,l,t} (-1)^{n-l}
\omega^{n-2l}(1-\omega)^{2l}\alpha_{2n-2l,m}
\end{split}
\end{equation}
The condition $n - \frac m2 \leq l\leq\lfloor\frac{n}{2}\rfloor$ implies $0 
\leq n \leq m$. Hence, we have
\begin{equation}
\begin{split}
\mu_m(t) &= \sum_{n=0}^m \sum_{l=n-\frac{m}{2}}^{\lfloor\frac{n}{2}\rfloor} 
C_{n,l,t} (-1)^{n-l}
\omega^{n-2l}(1-\omega)^{2l}\alpha_{2n-2l,m}\\
&= \sum_{n=0}^m \sum_{l=n-\frac{m}{2}}^{\lfloor\frac{n}{2}\rfloor}  
\frac{(-1)^nt^n}{8^nn!}  8^{2l}(-1)^n 
\binom{n}{2l}\binom{2n-2l}{n-l} 
\omega^{n-2l}(1-\omega)^{2l}\alpha_{2n-2l,m}\\
&= \sum_{n=0}^{m}\beta_{n,\omega}t^n.\label{eq:mixed-moment-general}
\end{split}
\end{equation}
Let us calculate the leading term, $\beta_{m,\omega}$. Then we have 
$l\in\{\frac{m}{2}\}
$ and $n=m$ and $m$ is even.
\begin{equation}
\begin{split}
\beta_m &= \sum_{l=m-\frac{m}{2}}^{\lfloor\frac{m}{2}\rfloor} 
\frac{(-1)^m}{8^mm!}  8^{2l}(-1)^m 
\binom{m}{2l}\binom{2m-2l}{m-l} \omega^{m-2l}(1-\omega)^{2l}\alpha_{2m-2l,m}\\
&= \frac{1}{8^mm!}  8^{m}
\binom{m}{m}\binom{m}{\frac{m}{2}}(1-\omega)^{m} \alpha_{m,m}\\
&= \frac{1}{m!} \binom{m}{\frac{m}{2}} (1-\omega)^{m} 
(-1)^mm! =  \binom{m}{\frac{m}{2}} (1-\omega)^{m},
\end{split}
\end{equation}
where we used the fact, that $\alpha_{m,m} =(-1)^mm!$ by 
Lemma~\ref{lem:sumOfPowerBinomial}.
\end{proof}

\bibliographystyle{ieeetr}
\bibliography{Bibliografia}

\begin{thebibliography}{10}

\bibitem{gonis2001decoherence}
A.~Gonis and P.~E. Turchi, {\em Decoherence and its implications in quantum
  computation and information transfer}, vol.~182.
\newblock IOS press, 2001.

\bibitem{amin2009decoherence}
M.~H. Amin, D.~V. Averin, and J.~A. Nesteroff, ``Decoherence in adiabatic
  quantum computation,'' {\em Physical Review A}, vol.~79, no.~2, p.~022107,
  2009.

\bibitem{kossakowski1972quantum}
A.~Kossakowski, ``On quantum statistical mechanics of non-hamiltonian
  systems,'' {\em Reports on Mathematical Physics}, vol.~3, no.~4,
  pp.~247--274, 1972.

\bibitem{lindblad1976generators}
G.~Lindblad, ``On the generators of quantum dynamical semigroups,'' {\em
  Communications in Mathematical Physics}, vol.~48, no.~2, pp.~119--130, 1976.

\bibitem{gorini1976completely}
V.~Gorini, A.~Kossakowski, and E.~C.~G. Sudarshan, ``Completely positive
  dynamical semigroups of n-level systems,'' {\em Journal of Mathematical
  Physics}, vol.~17, no.~5, pp.~821--825, 1976.

\bibitem{whitfield2010quantum}
J.~D. Whitfield, C.~A. Rodr{\'\i}guez-Rosario, and A.~Aspuru-Guzik, ``Quantum
  stochastic walks: A generalization of classical random walks and quantum
  walks,'' {\em Physical Review A}, vol.~81, no.~2, p.~022323, 2010.

\bibitem{venegas2012quantum}
S.~E. Venegas-Andraca, ``Quantum walks: a comprehensive review,'' {\em Quantum
  Information Processing}, vol.~11, no.~5, pp.~1015--1106, 2012.

\bibitem{robens2015ideal}
C.~Robens, W.~Alt, D.~Meschede, C.~Emary, and A.~Alberti, ``Ideal negative
  measurements in quantum walks disprove theories based on classical
  trajectories,'' {\em Physical Review X}, vol.~5, no.~1, p.~011003, 2015.

\bibitem{stefanak2010interference}
M.~Stefanak, ``Interference phenomena in quantum information,'' {\em arXiv
  preprint arXiv:1009.0200}, 2010.

\bibitem{vieira2013dynamically}
R.~Vieira, E.~P. Amorim, and G.~Rigolin, ``Dynamically disordered quantum walk
  as a maximal entanglement generator,'' {\em Physical review letters},
  vol.~111, no.~18, p.~180503, 2013.

\bibitem{kurzynski2013quantum}
P.~Kurzy{\'n}ski and A.~W{\'o}jcik, ``Quantum walk as a generalized measuring
  device,'' {\em Physical review letters}, vol.~110, no.~20, p.~200404, 2013.

\bibitem{kollar2014entropy}
B.~Koll{\'a}r and M.~Koniorczyk, ``Entropy rate of message sources driven by
  quantum walks,'' {\em Physical Review A}, vol.~89, no.~2, p.~022338, 2014.

\bibitem{kitagawa2010exploring}
T.~Kitagawa, M.~S. Rudner, E.~Berg, and E.~Demler, ``Exploring topological
  phases with quantum walks,'' {\em Physical Review A}, vol.~82, no.~3,
  p.~033429, 2010.

\bibitem{arnault2016quantum}
P.~Arnault and F.~Debbasch, ``Quantum walks and discrete gauge theories,'' {\em
  Physical Review A}, vol.~93, no.~5, p.~052301, 2016.

\bibitem{chandrashekar2010relationship}
C.~Chandrashekar, S.~Banerjee, and R.~Srikanth, ``Relationship between quantum
  walks and relativistic quantum mechanics,'' {\em Physical Review A}, vol.~81,
  no.~6, p.~062340, 2010.

\bibitem{chia2016coherent}
A.~Chia, A.~G{\'o}recka, T.~K. C, {\L}.~Pawela, K.~P, T.~Paterek, and
  D.~Kaszlikowski, ``Coherent chemical kinetics as quantum walks i: Reaction
  operators for radical pairs,'' {\em Physical Review E}, vol.~93, p.~032407,
  2016.

\bibitem{metzler2014anomalous}
R.~Metzler, J.-H. Jeon, A.~G. Cherstvy, and E.~Barkai, ``Anomalous diffusion
  models and their properties: non-stationarity, non-ergodicity, and ageing at
  the centenary of single particle tracking,'' {\em Physical Chemistry Chemical
  Physics}, vol.~16, no.~44, pp.~24128--24164, 2014.

\bibitem{spiechowicz2016transient}
J.~Spiechowicz, J.~{\L}uczka, and P.~H{\"a}nggi, ``Transient anomalous
  diffusion in periodic systems: ergodicity, symmetry breaking and velocity
  relaxation,'' {\em Scientific Reports}, vol.~6, 2016.

\bibitem{spiechowicz2015diffusion}
J.~Spiechowicz and J.~{\L}uczka, ``Diffusion anomalies in ac-driven brownian
  ratchets,'' {\em Physical Review E}, vol.~91, no.~6, p.~062104, 2015.

\bibitem{klages2008anomalous}
R.~Klages, G.~Radons, and I.~M. Sokolov, {\em Anomalous transport: foundations
  and applications}.
\newblock John Wiley \& Sons, 2008.

\bibitem{kempe2003quantum}
J.~Kempe, ``Quantum random walks: an introductory overview,'' {\em Contemporary
  Physics}, vol.~44, no.~4, pp.~307--327, 2003.

\bibitem{caldeira1983path}
A.~O. Caldeira and A.~J. Leggett, ``Path integral approach to quantum brownian
  motion,'' {\em Physica A: Statistical mechanics and its Applications},
  vol.~121, no.~3, pp.~587--616, 1983.

\bibitem{prokof2006decoherence}
N.~Prokof’ev and P.~Stamp, ``Decoherence and quantum walks: Anomalous
  diffusion and ballistic tails,'' {\em Physical Review A}, vol.~74, no.~2,
  p.~020102, 2006.

\bibitem{jacobs2006straightforward}
K.~Jacobs and D.~A. Steck, ``A straightforward introduction to continuous
  quantum measurement,'' {\em Contemporary Physics}, vol.~47, no.~5,
  pp.~279--303, 2006.

\bibitem{pasquini2006tridiagonal}
S.~Noschese, L.~Pasquini, and L.~Reichel, ``Tridiagonal toeplitz matrices:
  properties and novel applications,'' {\em Numerical linear algebra with
  applications}, vol.~20, no.~2, pp.~302--326, 2013.

\bibitem{havlin1982fractal}
S.~Havlin and D.~Ben-Avraham, ``Fractal dimensionality of polymer chains,''
  {\em Journal of Physics A: Mathematical and General}, vol.~15, no.~6,
  p.~L311, 1982.

\bibitem{haus1987diffusion}
J.~Haus and K.~Kehr, ``Diffusion in regular and disordered lattices,'' {\em
  Physics Reports}, vol.~150, no.~5, pp.~263--406, 1987.

\bibitem{havlin1987diffusion}
S.~Havlin and D.~Ben-Avraham, ``Diffusion in disordered media,'' {\em Advances
  in Physics}, vol.~36, no.~6, pp.~695--798, 1987.

\bibitem{grimmett2004weak}
G.~Grimmett, S.~Janson, and P.~F. Scudo, ``Weak limits for quantum random
  walks,'' {\em Physical Review E}, vol.~69, no.~2, p.~026119, 2004.

\bibitem{faccin2013degree}
M.~Faccin, T.~Johnson, J.~Biamonte, S.~Kais, and P.~Migda{\l}, ``Degree
  distribution in quantum walks on complex networks,'' {\em Physical Review X},
  vol.~3, no.~4, p.~041007, 2013.

\bibitem{gradshteyn2014table}
I.~S. Gradshteyn and I.~M. Ryzhik, {\em Table of integrals, series, and
  products}.
\newblock Academic press, 2014.

\end{thebibliography}

\end{document}